\documentclass[11pt]{article}
\usepackage{authblk}
\usepackage{amssymb}
\usepackage{amsthm}
\usepackage{mathtools}
\usepackage{amsfonts}
\usepackage{amsmath}
\usepackage{graphicx} 
\usepackage{amsmath,amsthm,amssymb}
\usepackage[margin=1in]{geometry}
\usepackage[colorlinks=true]{hyperref}
\usepackage{thmtools, thm-restate}
\usepackage{algorithm}
\usepackage{algorithm}
\usepackage{algpseudocode}

\declaretheorem{theorem}
\usepackage{float}

\newtheorem{lemma}{Lemma} 
 
\newtheorem{corollary}{Corollary}
\newtheorem{definition}{Definition}
\newtheorem{question}{Question}

\newtheorem{proposition}{Proposition}

\newtheorem*{definition*}{Definition}

\newcommand{\seq}{\operatorname{seq}}

\newcommand{\N}{\mathbb{N}}
\newcommand{\Q}{\mathbb{Q}}

\newcommand{\abs}[1]{\left|#1\right|}
\newcommand{\len}[1]{\left|#1\right|}
\newcommand{\prefix}{\!\upharpoonright\!}

\usepackage{tikz}
\usetikzlibrary{automata, positioning}
\usepackage{graphicx}
\title{On Normality and Equidistribution for Separator Enumerators}

\author{Subin Pulari}
\affil[1]{
National Research University Higher School of Economics Moscow
}

\date {\today}

\begin{document}

\maketitle

\begin{abstract}
	
A separator is a countable dense subset of $[0,1)$, and a separator enumerator is a naming scheme that assigns a real number in $[0,1)$ to each finite word so that the set of all named values is a separator. Mayordomo introduced separator enumerators to define $f$-normality and a relativized finite-state dimension $\dim^{f}_{\mathrm{FS}}(x)$, where finite-state dimension measures the asymptotic lower rate of finite-state information needed to approximate $x$ through its $f$-names. This framework extends classical base-$k$ normality, and Mayordomo showed that it supports a point-to-set principle for finite-state dimension. This representation-based viewpoint has since been developed further in follow-up work, including by Calvert et al., yielding strengthened randomness notions such as supernormal and highly normal numbers.

Mayordomo posed the following open question: can $f$-normality be characterized via equidistribution properties of the sequence $\left(|\Sigma|^{n} a^{f}_{n}(x)\right)_{n=0}^{\infty}$, where $a^{f}_{n}(x)$ is the sequence of best approximations to $x$ from below induced by $f$? We give a strong negative answer: we construct computable separator enumerators $f_0,f_1$ and a point $x$ such that $a^{f_0}_{n}(x)=a^{f_1}_{n}(x)$ for all $n$, yet $\dim^{f_0}_{\mathrm{FS}}(x)=0$ while $\dim^{f_1}_{\mathrm{FS}}(x)=1$. Consequently, no criterion depending only on the sequence $\left(|\Sigma|^{n} a^{f}_{n}(x)\right)_{n=0}^{\infty}$ - in particular, no equidistribution property of this sequence - can characterize $f$-normality uniformly over all separator enumerators. On the other hand, for a natural finite-state coherent class of separator enumerators we recover a complete equidistribution characterization of $f$-normality. We also show that beyond finite-state coherence, this characterization can fail even for a separator enumerator computable in nearly linear time.

\end{abstract}
\section{Introduction}

Finite-state dimension is a quantitative notion of the rate of randomness in an individual infinite sequence, as measured by finite automata.  It was introduced by Dai, Lathrop, Lutz, and Mayordomo~\cite{Dai2004} as a finite-state analogue of effective Hausdorff dimension~\cite{Lutz2003a,Lutz2003b}.  It admits several equivalent characterizations, including formulations via finite-state gambling, finite-state compression, and block entropy rates~\cite{Dai2004,Doty2006,Bourke2005}.  It also has an important connection with the theory of normal numbers: a real number is normal to base $k$ if and only if its base-$k$ digit sequence has finite-state dimension equal to $1$~\cite{Bourke2005}.

An important connection between effective dimension and classical fractal dimension is provided by point-to-set principles.
Lutz and Lutz~\cite{Lutz2018} proved that the Hausdorff dimension in Euclidean spaces of a set can be obtained by minimizing,
over oracles, the supremum of the relativized effective dimensions of its points.
This reduces many lower-bound questions in geometric measure theory to analyzing the information density of carefully chosen
points, and it has led to several new results and new proofs across classical fractal geometry; see~\cite{Lutz2020}.
In a recent work, Mayordomo~\cite{Mayordomo2025} established an analogous point-to-set principle for finite-state dimension.  A key feature of this development is that it is representation-based: instead of fixing a base-$k$ expansion, it uses \emph{separator enumerators}, i.e.\ naming schemes that assign reals in $[0,1)$ to finite words so that the range forms a countable dense subset.  With such an enumerator $f:\Sigma^*\to[0,1)$, one can measure how efficiently a finite-state transducer can output an $f$-name approximating a real $x$ to a given precision, leading to a relativized finite-state dimension $\dim^f_{\mathrm{FS}}(x)$ and the induced notion of \emph{$f$-normality}, defined by the condition $\dim^f_{\mathrm{FS}}(x)=1$~\cite{Mayordomo2025}. Calvert et al.~\cite{Calvert2025} develop this framework further, introducing strengthened notions such as supernormal and highly normal numbers under broad classes of representations.

In this paper we study the relationships between $f$-normality and equidistribution.  
In the classical base-$k$ setting, normality has a sharp equidistribution characterization: $x$ is normal to base $k$ if and only if the sequence $(k^n x)_{n\ge1}$ is uniformly distributed modulo $1$~\cite{KuipersNiederreiterUniform}. Motivated by the classical equivalence between normality and equidistribution, Mayordomo~\cite{Mayordomo2025} asked whether an analogous equidistribution criterion holds in the setting of $f$-normality.  Let $(a_n^f(x))_{n\ge 0}$ denote the best-approximation-from-below sequence associated with $f$ and $x$, i.e.\ $a_n^f(x)=\max\{f(w): |w|\le n,\ f(w)\le x\}$.  Mayordomo posed the following open question:
\medskip
\par\noindent\emph{Can $f$-normality be characterized via equidistribution properties of the sequence
	$(|\Sigma|^{n}\!a^{f}_{n}(x))_{n\ge 0}$?}
\medskip

We show that the answer is negative in a strong way.  We construct two total computable rational-valued separator enumerators $f_0,f_1$ and a point $x\in[0,1)$ such that
$a_n^{f_0}(x)=a_n^{f_1}(x)$ for all $n$, yet $\dim^{f_0}_{\mathrm{FS}}(x)=0$ while $\dim^{f_1}_{\mathrm{FS}}(x)=1$.
Consequently, no criterion depending only on the single numeric sequence $\bigl(|\Sigma|^n a_n^f(x)\bigr)_{n\ge 0}$---in particular, no equidistribution property of that sequence---can characterize $f$-normality uniformly over all separator enumerators.

At the same time, an equidistribution characterization does hold under a natural structural restriction on the representation.  We identify a class of \emph{finite-state coherent} enumerators obtained from the standard base-$k$ grid by an invertible synchronous Mealy-machine relabeling, and we prove that in this regime $f$-normality is equivalent to a $k$-adic equidistribution property of the integer sequence $\bigl(k^n a_n^f(x)\bigr)_{n\ge 1}$ (uniform distribution of residues modulo $k^m$ for every fixed $m$). Beyond finite-state coherence, we show that this equivalence can already break for a separator enumerator $f$ computable in $O(n\log n)$ time: there is a point $x$ whose scaled best-from-below approximations are $k$-adically equidistributed, but $x$ is not $f$-normal.

Section~\ref{sec:preliminaries} defines separator enumerators, finite-state transducers, $\dim^f_{\mathrm{FS}}$, and $f$-normality. Section \ref{sec:noeqcharacterization} proves the negative result via a pair of computable enumerators with identical best-from-below chains but sharply different relativized finite-state dimension.  Section \ref{sec:coherentenumerators} establishes the $k$-adic equidistribution characterization for finite-state coherent enumerators, and we then show that this correspondence can already fail for a separator enumerator computable in nearly linear time. 



\section{Preliminaries}
\label{sec:preliminaries}

Let $\Sigma$ be a finite alphabet with $k=\len{\Sigma}\ge 2$ and write $\Sigma^{*}$ for the set of finite strings over $\Sigma$.
For $w\in\Sigma^*$, $\len{w}$ denotes its length, and for an infinite sequence $X\in\Sigma^\infty$ we write
$X\prefix n$ for its length-$n$ prefix (using the operator $\prefix$ defined in the preamble).
We use $\lfloor\cdot\rfloor$ for the floor function and interpret congruences $b_n\equiv r\pmod{k^m}$ in the usual sense.
Throughout this section we fix an identification $\Sigma=\{0,1,\dots,k-1\}$. We view finite words over $\Sigma$ as base-$k$ numerals and use the associated $k$-adic grid in $[0,1)$.
For a word $u=u_1u_2\cdots u_n\in \Sigma^n$, define its base-$k$ value
\[
\mathrm{val}(u)=\sum_{i=1}^n u_i\,k^{n-i}\in\{0,1,\dots,k^n-1\},
\qquad\text{and}\qquad
\mathrm{grid}(u)=\frac{\mathrm{val}(u)}{k^n}\in[0,1).
\]
Thus $\mathrm{grid}(\Sigma^n)=\{j/k^n:0\le j<k^n\}$.
Let $\seq_k(x)\in\Sigma^\infty$ denote the (canonical) base-$k$ expansion of $x\in[0,1)$ chosen so as \emph{not}
to end in $(k-1)^\infty$.


Finite-state transducers model one-pass, constant-memory transformations on words and are used to measure finite-state description length.

\begin{definition}[Finite-state transducer (FST)]
	A \emph{$\Sigma$-finite-state transducer} (briefly, \emph{$\Sigma$-FST}) is a tuple
	$T=(Q,\delta,\nu,q_0)$ where $Q$ is a finite nonempty set of states, $\delta:Q\times\Sigma\to Q$ is a transition
	function, $\nu:Q\times\Sigma\to\Sigma^{*}$ is an output function, and $q_0\in Q$ is the start state.
	
	The transition function extends to words by $\delta(q,\lambda)=q$ and $\delta(q,wa)=\delta(\delta(q,w),a)$.
	For $q\in Q$ and $w\in\Sigma^{*}$, define the output $\nu(q,w)\in\Sigma^{*}$ by $\nu(q,\lambda)=\lambda$ and
	$\nu(q,wa)=\nu(q,w)\,\nu(\delta(q,w),a)$ for $a\in\Sigma$.
	The overall output of $T$ on input $w$ is $T(w)=\nu(q_0,w)$.
\end{definition}

We measure how concisely a transducer can generate a given target string.

\begin{definition}[$T$-information content \cite{Mayordomo2025}]
	Let $T$ be a $\Sigma$-FST and $w\in\Sigma^{*}$. Define
	$$
	K^{T}(w)=\min\{\len{\pi}:\ \pi\in\Sigma^{*}\ \text{and}\ T(\pi)=w\},
	$$
	with the convention $K^T(w)=\infty$ if $w$ is not in the range of $T$.
\end{definition}


A separator enumerator is a naming scheme for a countable dense subset of $[0,1)$, assigning a real in $[0,1)$ to each finite word.

\begin{definition}[Separator, separator enumerator \cite{Mayordomo2025}]
	A set $S\subseteq[0,1)$ is a \emph{separator} if it is countable and dense in $[0,1)$. A function
	$f:\Sigma^{*}\to[0,1)$ is a \emph{separator enumerator} (SE) if $\mathrm{Im}(f)$ is a separator.
\end{definition}

We call an SE $f:\Sigma^*\to[0,1)$ \emph{total computable} if there is an algorithm that, on input $(w,t)$ with $w\in\Sigma^*$ and $t\in\N$, outputs a rational $q$ such that $|q-f(w)|\le 2^{-t}$, and halts on every input.
Given an enumerator $f$, we quantify how many input symbols a transducer needs in order to produce some $f$-name that $\delta$-approximates a target point.

\begin{definition}[Relativized approximation complexity \cite{Mayordomo2025}]
	Let $f$ be an SE, $T$ a $\Sigma$-FST, $\delta>0$, and $x\in[0,1)$. Define
	$$
	K^{T,f}_{\delta}(x)=\min\{K^{T}(w):\ w\in\Sigma^{*}\ \text{and}\ \abs{f(w)-x}<\delta\}.
	$$
\end{definition}

The induced finite-state dimension is the optimal asymptotic approximation rate achievable by finite-state transducers.

\begin{definition}[Relativized finite-state dimension and $f$-normality \cite{Mayordomo2025}]
	Let $f$ be an SE and $x\in[0,1)$. Define
	\[
	\dim_{\mathrm{FS}}^{f}(x)=\inf_{T\ \Sigma\text{-FST}}\ \liminf_{\delta\to 0^{+}}\ \frac{K^{T,f}_{\delta}(x)}{\log_{k}(1/\delta)}.
	\]
	We say $x$ is \emph{$f$-normal} if $\dim_{\mathrm{FS}}^{f}(x)=1$.
\end{definition}

The standard base-$k$ naming map will serve as the baseline enumerator throughout. Let $f_{\mathrm{std}}:\Sigma^*\to[0,1)$ be the standard base-$k$ enumerator
\[
f_{\mathrm{std}}(u)=\mathrm{grid}(u)=\frac{\mathrm{val}(u)}{k^{\len{u}}}\qquad(u\neq\lambda),
\]
with $f_{\mathrm{std}}(\lambda)=0$. We now define the induced \emph{best approximation from below} chain associated with a separator enumerator.

\begin{definition}[Best approximation from below \cite{Mayordomo2025}]
	\label{def:bestapprox}
	Let $f$ be an SE and $x\in[0,1)$. For each $n\in\N$, define $a_n^{f}(x)$ to be any value $f(w)$ with $\len{w}\le n$
	such that $f(w)\le x$ and $x-f(w)$ is minimum among all $u$ with $\len{u}\le n$ and $f(u)\le x$.
	Equivalently,
	$$
	a_n^{f}(x)=\max\{f(w):\ \len{w}\le n,\ f(w)\le x\}.
	$$
\end{definition}


\section{Equidistribution cannot characterize normality for separator enumerators}
\label{sec:noeqcharacterization}

We now address Mayordomo's question on equidistribution criteria for $f$-normality. Fix a separator enumerator $f:\Sigma^*\to[0,1)$ and $x\in[0,1)$, and let $(a_n^f(x))_{n\ge 0}$ be the best approximation from below sequence
corresponding to $f$ and $x$ (see Definition \ref{def:bestapprox}).

\begin{question}[\cite{Mayordomo2025}]
	Can $f$-normality be characterized via equidistribution properties of the sequence
	$(|\Sigma|^{n}\,a^{f}_{n}(x))_{n\ge 0}$?
\end{question}


We answer this question in the negative by constructing two total computable rational-valued separator enumerators $f_0,f_1$ and a point $x\in[0,1)$
such that $a_n^{f_0}(x)=a_n^{f_1}(x)$ for all $n$, yet $\dim^{f_0}_{\mathrm{FS}}(x)=0$ while $\dim^{f_1}_{\mathrm{FS}}(x)=1$.

\subsection{Preliminary lemmas}

From the finite-state transducer characterization of finite-state dimension due to Doty and Moser (see Theorem 3.11 from \cite{Doty2006}), we obtain the following corollary.

\begin{lemma}\label{lem:normal-dim1}
	There exists an infinite sequence $Z\in\Sigma^{\infty}$ such that for every $\Sigma$-FST $T$,
	$$\liminf_{n\to\infty}\frac{K^{T}(z_n)}{n}=1,$$ where $z_n:=Z\prefix n$.
\end{lemma}
\begin{proof}
	Fix any $\Sigma$-normal sequence $Z$. By the transducer characterization of finite-state dimension and the fact that normal sequences have finite-state dimension equal to $1$,
	\[
	1=\dim_{\mathrm{FS}}(Z)=\inf_{T}\ \liminf_{n\to\infty}\frac{K^{T}(z_n)}{n}.
	\]
	Therefore, for every $\Sigma$-FST $T$ one must have $\liminf_{n} K^{T}(z_n)/n = 1$; otherwise the infimum would be $<1$.
\end{proof}

We also need a second sequence whose prefixes are distinct from the first but have the same property.
This is achieved by a fixed symbol permutation.

\begin{lemma}\label{lem:perm}
	Let $\pi:\Sigma\to\Sigma$ be a bijection and extend it letterwise to $\pi:\Sigma^{*}\to\Sigma^{*}$.
	If $Z\in\Sigma^{\infty}$ satisfies Lemma~\ref{lem:normal-dim1}, then so does $\pi(Z)$, i.e.\ the prefixes
	$t_n := \pi(Z)\prefix n = \pi(z_n)$ satisfy
	$\liminf_{n}K^{T}(t_n)/n=1$ for every $\Sigma$-FST $T$.
\end{lemma}
\begin{proof}
	Fix a $\Sigma$-FST $T$. Let $P$ be the 1-state transducer that maps each input symbol $a$ to the single output
	symbol $\pi(a)$. Then $P(w)=\pi(w)$ for all $w$, and similarly there is a 1-state transducer $P^{-1}$ with
	$P^{-1}(w)=\pi^{-1}(w)$ for all $w$. For any string $u$,
	\[
	K^{T}(\pi(u)) = K^{T}(P(u)) \ge K^{P^{-1}\circ T} (u),
	\]
	because any input producing $\pi(u)$ under $T$ yields an input producing $u$ under $P^{-1}\circ T$.
	Thus, for $t_n=\pi(z_n)$,
	\[
	\frac{K^{T}(t_n)}{n}\ \ge\ \frac{K^{P^{-1}\circ T}(z_n)}{n}.
	\]
	Taking $\liminf$ and applying Lemma~\ref{lem:normal-dim1} for the transducer $P^{-1}\circ T$ we obtain
	$\liminf_{n}K^{T}(t_n)/n = 1$. 
\end{proof}

Fix a bijection $\pi:\Sigma\to\Sigma$ with no fixed points (such a derangement exists for all $k\ge 2$) and with $\pi(0)\neq 0$.
Let $Z$ be as in Lemma~\ref{lem:normal-dim1}$,$ let $Y=\pi(Z)$, and write
$z_n:=Z\prefix n$ and $y_n:=Y\prefix n=\pi(z_n)$. Then $z_n\neq y_n$ for all $n$ because $\pi$ has no fixed points, and
$y_n\neq 0^n$ for all $n$ because $\pi(0)\neq 0$.

\subsection{Construction of the two separator enumerators}

Fix $x:=\tfrac12\in(0,1)$. For each $n\in\N$, define $r_n:=k^{-(n+2)}$ and
\[
A_n:=\bigl(x-r_n,\ x-r_{n+1}\bigr)\ \cup\ \bigl(x+r_{n+1},\ x+r_n\bigr).
\]
Then $A_n$ is a nonempty open set and $(A_n)_{n\ge 0}$ partition the punctured neighborhood
$$\bigl(x-r_0,\ x+r_0\bigr)\setminus\{x\}=\bigsqcup_{n\ge 0} A_n$$, and $\log_k(1/r_n)=n+2$.
Define the target \emph{best-from-below} values
$$c_n:=x-\frac{r_n+r_{n+1}}{2}\in(x-r_n,x-r_{n+1})\subseteq A_n$$. These values will later be enforced as the canonical
length-$n$ approximants to $x$ for both enumerators, i.e., we will arrange $a_n^{f_i}(x)=c_n$ for each $n$.
Observe that $c_n<x$ and $c_n\uparrow x$ as $n\to\infty$.

For each $n$, fix a computable listing $(d_{n,t})_{t\in\N}$ of \emph{distinct} rationals in
$A_n\cap\Q\setminus\{c_n\}$ whose range is dense in $A_n$, and set $D_n:=\{d_{n,t}:t\in\N\}$.
Also fix a computable listing $(q_t)_{t\in\N}$ of \emph{distinct} rationals in
$\bigl([0,1)\setminus(x-r_0,x+r_0)\bigr)\cap\Q$ whose range is dense there, and set
$D_{\mathrm{far}}:=\{q_t:t\in\N\}$. We now define two functions $f_0,f_1:\Sigma^{*}\to[0,1)$.

\noindent\textbf{Step 1 (fixing the best-from-below chain for lengths $\le n$).}
For each $n\ge 0$, define
\[
f_0(0^n) := c_n,\qquad f_1(z_n) := c_n.
\]
(Here $0\in\Sigma$ is a fixed symbol, and $0^n$ is the all-$0$ word of length $n$.)

\noindent\textbf{Step 2 (ensuring density near $x$ using incompressible prefixes).}
Partition $\N$ into pairwise disjoint infinite sets $(L_n)_{n\ge 0}$ such that $m\in L_n$ implies $m\ge n+1$.
For concreteness, one may take $L_n=\{2^{n}(2t+1): t\in\N\}$; then $m\in L_n$ implies $m\ge 2^n\ge n+1$ for $n\ge 0$.
Fix computable bijections $\varphi_n:L_n\to D_n$ by setting, for $t\in\N$,
$$
\varphi_n\bigl(2^{n}(2t+1)\bigr):=d_{n,t}.
$$
Now set, for each $n\ge 0$ and each $m\in L_n$,
\[
f_0(y_m)=f_1(y_m):=\varphi_n(m)\in A_n.
\]
Note that there is no conflict with Step~1, since $y_m\neq 0^m$ and $y_m\neq z_m$ for all $m$.
Thus, for each $n$, the set $\{f_i(y_m): m\in L_n\}$ is dense in $A_n$ (for both $i=0,1$).

\noindent\textbf{Step 3 (defining all remaining values far from $x$ while maintaining density).}
Let $W$ be the set of all remaining strings not yet assigned a value by Steps 1--2:
\[
W := \Sigma^{*}\setminus\Bigl(\{0^n:n\in\N\}\ \cup\ \{z_n:n\in\N\}\ \cup\ \{y_n:n\in\N\}\Bigr).
\]
Enumerate $\Sigma^{*}$ in the length lexicographic order, then all words of length $2$ in lexicographic order, and so on.  Call this enumeration
$(u_j)_{j\in\N}$. Define $w_j$ to be the $j$th word in this list that lies in $W$, i.e., the $j$th $u_t$ such that
$
u_t\notin \{0^n:n\in\N\}\ \cup\ \{z_n:n\in\N\}\ \cup\ \{y_n:n\in\N\}.
$
Also enumerate $D_{\mathrm{far}}=\{q_0,q_1,q_2,\dots\}$ according to the fixed computable listing above.
Define $f_0(w_j)=f_1(w_j):=q_j$ for all $j$.

\begin{lemma}\label{lem:SE}
	The functions $f_0,f_1$ defined above are total computable rational-valued separator enumerators:
	each $\mathrm{Im}(f_i)$ is countable and dense in $[0,1)$.
\end{lemma}

\begin{proof}
	Countability is immediate since $\Sigma^{*}$ is countable.
	
	We also note that $f_0$ and $f_1$ are total computable (with rational outputs).
	Assume the fixed dense sets $D_n$ and $D_{\mathrm{far}}$ come with fixed computable enumerations, and that each $\varphi_n:L_n\to D_n\setminus\{c_n\}$
	is a fixed computable bijection. Also fix $Z$ to be a computable $\Sigma$-normal sequence, so $n\mapsto z_n=Z\prefix n$ is computable; then $Y=\pi(Z)$ is computable and $n\mapsto y_n=Y\prefix n$ is computable.
	
	On input $w\in\Sigma^{*}$ of length $\ell$, we compute $f_i(w)$ as follows.
	First check whether $w=0^\ell$ (all symbols equal $0$), whether $w=z_\ell$ (compute $z_\ell$ and compare), and whether $w=y_\ell$ (compute $y_\ell$ and compare).
	If $w=0^\ell$ output $f_0(w)=c_\ell$; if $w=z_\ell$ output $f_1(w)=c_\ell$; and if $w=y_\ell$ then find the unique $n$ with $\ell\in L_n$ (the $L_n$ are decidable and disjoint) and output $f_0(w)=f_1(w)=\varphi_n(\ell)$.
	Otherwise $w\in W$, and $W$ is decidable because $w\notin W$ iff one of the three checks above holds.
	In this case, compute the index $j$ such that $w=w_j$ in Step~3 by enumerating all words of $\Sigma^{*}$ by increasing length, and within each fixed length in lexicographic order, skipping exactly those words that are not in $W$, until $w$ is reached; then output $f_0(w)=f_1(w)=q_j$ (where $(q_j)$ is the fixed computable enumeration of $D_{\mathrm{far}}$).
	This procedure halts for every input $w$, so $f_0,f_1$ are total computable.
	
	For density, let $I\subseteq[0,1)$ be a nonempty open interval.
	If $I$ intersects $[0,1)\setminus(x-r_0,x+r_0)$, then $I$ contains a rational in $D_{\mathrm{far}}$, hence an image point of $f_i$.
	Otherwise, $I\subseteq(x-r_0,x+r_0)$, so $I$ intersects $A_n$ for some $n\ge 0$ (because the annuli partition
	$(x-r_0,x+r_0)\setminus\{x\}$ and $I$ is open, hence cannot be $\{x\}$). Since $D_n$ is dense in $A_n$ and
	$\{f_i(y_m):m\in L_n\}=\varphi_n(L_n)$ is dense in $A_n$, the interval $I$ contains some $f_i(y_m)$.
	Thus $\mathrm{Im}(f_i)$ meets every nonempty open interval, so it is dense.
\end{proof}

Since $f_0$ and $f_1$ are rational-valued and the above procedure computes the defining case and the corresponding index effectively, it in fact yields exact computation: there is a Turing machine that, on input $w\in\Sigma^*$, halts and outputs the rational value $f_i(w)$ itself (rather than merely producing $2^{-t}$-approximations). The next lemma shows that $f_0$ and $f_1$ induce exactly the same sequence of length-bounded approximations of $x$ from below (and hence the same associated numeric scaling sequence).

\begin{lemma}\label{lem:an-same}
	For every $n\in\N$,
	$
	a_n^{f_0}(x)=a_n^{f_1}(x)=c_n.
	$
	Consequently, the numeric sequences $(k^{n}a_n^{f_0}(x))$ and $(k^{n}a_n^{f_1}(x))$ are identical.
\end{lemma}

\begin{proof}
	Fix $n$. We first show $a_n^{f_0}(x)=c_n$. By definition, $f_0(0^n)=c_n\le x$, so $a_n^{f_0}(x)\ge c_n$.
	Now consider any string $w$ with $\len{w}\le n$ and $f_0(w)\le x$. If $w=0^m$ for some $m\le n$, then $f_0(w)=c_m\le c_n$ because $(c_m)$ is increasing. If $w=y_m$ for some $m$, then $f_0(w)\in A_j$ for some $j$ with $m\in L_j$. In particular, $f_0(w)\le x-r_{j+1}<x$.
	Moreover, since $m\in L_j$ implies $m\ge j+1$, we have $j\le m-1\le n-1$ whenever $m\le n$.
	Thus $f_0(w)\le x-r_{j+1}\le x-r_n < c_n$ (because $c_n>x-r_n$ by construction). Finally, if $w\in W$ then $f_0(w)\in D_{\mathrm{far}} \subseteq [0,1)\setminus(x-r_0,x+r_0)$, so either $f_0(w)\le x-r_0<x-r_n<c_n$
	or $f_0(w)\ge x+r_0>x$. In either case it cannot exceed $c_n$ while staying $\le x$. Therefore, among all $\len{w}\le n$ with $f_0(w)\le x$, the maximum is attained at $w=0^n$ with value $c_n$.
	Hence $a_n^{f_0}(x)=c_n$.
	
	The proof for $f_1$ is identical, replacing the witness $0^n$ by $z_n$ (since $f_1(z_n)=c_n$) and noting that no other string of length
	$\le n$ attains a value in $(c_n,x]$ by the same case analysis.
	Thus $a_n^{f_1}(x)=c_n$.
\end{proof}

We now complete the construction by showing that, for the fixed point $x=\tfrac12$, the two separator enumerators $f_0$ and $f_1$ constructed above induce different relativized finite-state dimensions: $\dim_{\mathrm{FS}}^{f_0}(x)=0$ while $\dim_{\mathrm{FS}}^{f_1}(x)=1$. Together with Lemma~\ref{lem:an-same}, this yields a negative answer to Mayordomo's open question \cite{Mayordomo2025}.

\begin{theorem}\label{thm:negative}
	There exist total computable rational-valued separator enumerators $f_0,f_1:\Sigma^{*}\to[0,1)$ and a point $x\in[0,1)$ such that
	\[
	a_n^{f_0}(x)=a_n^{f_1}(x)\ \ \text{for all }n,
	\]
	yet
	\[
	\dim_{\mathrm{FS}}^{f_0}(x)=0
	\qquad\text{and}\qquad
	\dim_{\mathrm{FS}}^{f_1}(x)=1.
	\]
	In particular, $x$ is not $f_0$-normal, while $x$ is $f_1$-normal.
\end{theorem}

\begin{proof}
	We use the constructions above with the fixed $x=\tfrac12$.
	By Lemma~\ref{lem:an-same}, the sequences $a_n^{f_0}(x)$ and $a_n^{f_1}(x)$ coincide (both equal $c_n$),
	so $(k^n a_n(x))$ is identical for $f_0$ and $f_1$. It remains to compute the relativized finite-state dimensions.
	
	\noindent
	\textbf{Part 1: We show $\dim_{\mathrm{FS}}^{f_0}(x)=0$.}
	For each integer $L\ge 1$, let $T_L$ be the 1-state transducer that outputs $0^{L}$ on every input symbol.
	Then $T_L(\pi)=0^{L\len{\pi}}$ for all $\pi$, hence
	$$
	K^{T_L}(0^n)\le \left\lceil \frac{n}{L}\right\rceil.
	$$
	Fix $n$ and consider $\delta_n:=r_n/2$ (so $\log_k(1/\delta_n)=n+2+\log_k 2$).
	Since $f_0(0^n)=c_n$ and $\abs{x-c_n}=\abs{x-(x-(r_n+r_{n+1})/2)}=(r_n+r_{n+1})/2<r_n=\!2\delta_n$,
	we have $\abs{f_0(0^n)-x}<2\delta_n$.Therefore $$K^{T_L,f_0}_{2\delta_n}(x)\le K^{T_L}(0^n)\le \left\lceil \frac{n}{L}\right\rceil.$$ Hence
	\[
	\liminf_{n\to\infty}\frac{K^{T_L,f_0}_{2\delta_n}(x)}{\log_k(1/(2\delta_n))}
	\le
	\liminf_{n\to\infty}\frac{\lceil n/L\rceil}{(n+2)+O(1)}
	=\frac{1}{L}.
	\]
	To justify that this controls the full $\liminf_{\delta\to 0^+}$ (and not only the subsequence $2\delta_n$), note that
	$K^{T_L,f_0}_{\delta}(x)$ is monotone non-increasing in $\delta$. Hence for any $\delta\in(2\delta_{n+1},\,2\delta_n]$,
	\[
	K^{T_L,f_0}_{\delta}(x)\ \le\ K^{T_L,f_0}_{2\delta_{n+1}}(x).
	\]
	Moreover, for such $\delta$ we have $\log_k(1/\delta)\ge \log_k(1/(2\delta_n))$. Therefore,
	\[
	\frac{K^{T_L,f_0}_{\delta}(x)}{\log_k(1/\delta)}
	\ \le\
	\frac{K^{T_L,f_0}_{2\delta_{n+1}}(x)}{\log_k(1/(2\delta_n))}.
	\]
	Taking $\liminf$ over all $\delta\to 0^+$ and using that $n\to\infty$ along the corresponding intervals yields
	\begin{align*}
		\liminf_{\delta\to 0^+}\frac{K^{T_L,f_0}_{\delta}(x)}{\log_k(1/\delta)}
		&\le\
		\liminf_{n\to\infty}\frac{K^{T_L,f_0}_{2\delta_{n+1}}(x)}{\log_k(1/(2\delta_{n+1}))} \cdot \frac{\log_k(1/(2\delta_{n+1}))}{\log_k(1/(2\delta_{n}))} \\
		&\le\
		\liminf_{n\to\infty}\frac{K^{T_L,f_0}_{2\delta_{n+1}}(x)}{\log_k(1/(2\delta_{n+1}))} \cdot \lim_{n\to\infty} \frac{n+3+O(1)}{n+2+O(1)} \\
		& \le\ \frac{1}{L}.
	\end{align*}
	Since the above holds for all $T_L$, taking the infimum over all transducers $T$ and then letting $L\to\infty$ yields
	$\dim_{\mathrm{FS}}^{f_0}(x)=0$.
	
	\noindent\textbf{Part 2: We show $\dim_{\mathrm{FS}}^{f_1}(x)=1$.}
	Fix an arbitrary $\Sigma$-FST $T$. We prove that $\liminf_{\delta\to 0^+}\frac{K^{T,f_1}_{\delta}(x)}{\log_k(1/\delta)}\ge 1$,
	which implies $\dim_{\mathrm{FS}}^{f_1}(x) = 1$ after taking the infimum over $T$.
	
	Consider the scale $\delta_n:=r_n$. Any value $f_1(w)$ within distance $\delta_n=r_n$ of $x$ must lie in
	$(x-r_n,\ x+r_n)=\{x\}\cup \bigsqcup_{j\ge n} A_j$, so $w$ must be either (i) one of the special words $z_j$ with $j\ge n$
	(since $f_1(z_j)=c_j\in A_j$), or (ii) one of the special words $y_m$ with $f_1(y_m)\in A_j$ for some $j\ge n$
	(these are the values placed densely in the annuli), because by construction all other words are mapped into
	$D_{\mathrm{far}}$ outside $(x-r_0,x+r_0)$. Therefore,
	\[
	K^{T,f_1}_{r_n}(x)\ \ge\ \min\Bigl(\ \min_{j\ge n} K^T(z_j),\ \min_{m:\ f_1(y_m)\in \bigcup_{j\ge n}A_j} K^T(y_m)\Bigr).
	\]
	We now lower bound each term asymptotically by $n$. By Lemma~\ref{lem:normal-dim1}, we have  $\liminf_{j\to\infty} K^T(z_j)/j=1$.
	Hence for every $\varepsilon>0$ there exists $J$ such that for all $j\ge J$, $K^T(z_j)\ge (1-\varepsilon)j$; consequently,
	for all $n\ge J$, $\min_{j\ge n}K^T(z_j)\ge (1-\varepsilon)n$. Similarly, by Lemma~\ref{lem:perm} applied to $Y$,
	$\liminf_{m\to\infty}K^T(y_m)/m=1$, so for every $\varepsilon>0$ there exists $M$ such that for all $m\ge M$,
	$K^T(y_m)\ge (1-\varepsilon)m$. In our construction, if $f_1(y_m)\in A_j$, then $m\in L_j$, and by design $m\ge j+1$;
	therefore, whenever $f_1(y_m)\in \bigcup_{j\ge n}A_j$, we have $m\ge n+1$. For all $n\ge M$,
	\[
	\min_{m:\ f_1(y_m)\in \bigcup_{j\ge n}A_j} K^T(y_m)\ \ge\ (1-\varepsilon)(n+1)\ \ge\ (1-\varepsilon)n.
	\]
	Putting the two bounds together, for all sufficiently large $n$, $K^{T,f_1}_{r_n}(x)\ge (1-\varepsilon)n$.
	Since $\log_k(1/r_n)=n+2$, we obtain
	\[
	\liminf_{n\to\infty}\frac{K^{T,f_1}_{r_n}(x)}{\log_k(1/r_n)}
	\ \ge\ \liminf_{n\to\infty}\frac{(1-\varepsilon)n}{n+2}\ =\ 1-\varepsilon.
	\]
	As $\varepsilon>0$ was arbitrary, this yields $\liminf_{n\to\infty}\frac{K^{T,f_1}_{r_n}(x)}{\log_k(1/r_n)}\ge 1$.
	By monotonicity of $K^{T,f_1}_{\delta}(x)$ in $\delta$ and the fact that any $\delta\in(r_{n+1},r_n]$ satisfies
	$\log_k(1/\delta)\in[n+2,n+3)$, it follows that the full $\liminf_{\delta\to 0^+}$ is also at least $1$. Thus $$\liminf_{\delta\to 0^+}\frac{K^{T,f_1}_{\delta}(x)}{\log_k(1/\delta)}\ge 1$$ for every $T$, which implies that
	$\dim_{\mathrm{FS}}^{f_1}(x)=1$.

\end{proof}

As an immediate consequence, $f$-normality cannot be characterized by any property of the single numerical sequence $(k^{n}a_n^{f}(x))$.

\begin{corollary}\label{cor:nochar}
	There is no property $P$ of the numeric sequence $(k^{n}a_n^{f}(x))_{n\in\N}$ (in particular, no equidistribution
	property of this sequence) such that for all separator enumerators $f$ and all $x\in[0,1)$,
	\[
	x\text{ is $f$-normal}\ \Longleftrightarrow\ (k^{n}a_n^{f}(x))\text{ has property }P.
	\]
\end{corollary}

\begin{proof}
	Take $f_0,f_1,x$ from Theorem~\ref{thm:negative}. By Lemma~\ref{lem:an-same}, the sequences $(k^{n}a_n^{f_0}(x))$
	and $(k^{n}a_n^{f_1}(x))$ are identical, so they either both satisfy $P$ or both fail $P$.
	But by Theorem~\ref{thm:negative}, $x$ is $f_1$-normal and not $f_0$-normal, so no such $P$ can exist.
\end{proof}

\section{Finite-State Coherent Enumerators and an Equidistribution Characterization of $f$-Normality}
\label{sec:coherentenumerators}

From the previous section, we know that no equidistribution (or other distributional) criterion depending only on the sequence $(k^{n}a_n^{f}(x))_{n\ge 1}$ can characterize $f$-normality uniformly over all separator
enumerators. The goal of this section is to isolate a natural structural regime in which such a characterization
\emph{does} hold. We do this by restricting to \emph{finite-state coherent enumerators}, i.e.\ naming schemes obtained
from the standard base-$k$ grid by an invertible synchronous Mealy-machine relabeling.



To state an equidistribution characterization in this setting, we must use a notion that remains meaningful for the standard enumerator. For the standard base-$k$ naming map
$f_{\mathrm{std}}(w)=\mathrm{val}(w)/k^{|w|}$, the scaled approximation sequence
$b_n(x):=k^n a_n^{f_{\mathrm{std}}}(x)$ is integer-valued (indeed $b_n(x)=\lfloor k^n x\rfloor$), so the usual notion of
equidistribution modulo $1$ becomes trivial. The natural replacement is to ask for uniform distribution of the residues
of $b_n(x)$ at every finite base-$k$ resolution, i.e.\ modulo $k^m$ for each fixed $m$. This leads to the following
notion of $k$-adic equidistribution, which we adopt in the rest of the paper.

\begin{definition}[$k$-adic equidistribution]\label{def:kadic}
	A sequence $(b_n)_{n\ge 1}$ of integers is \emph{$k$-adically equidistributed} if for every $m\ge 1$ and every residue
	$r\in\{0,1,\dots,k^m-1\}$,
	\[
	\lim_{N\to\infty}\frac{1}{N}\#\{1\le n\le N:\ b_n\equiv r\pmod{k^m}\}=\frac{1}{k^m}.
	\]
\end{definition}

For finite-state coherent enumerators, the best-from-below approximants admit an explicit closed form, and the induced
$f$-dimension is invariant under finite-state coherent relabelings. These two facts combine to yield a clean
equidistribution characterization of $f$-normality in terms of $k$-adic equidistribution of the integer sequence
$(k^{n}a_n^{f}(x))_{n\ge 1}$.

\subsection{Finite-state coherent enumerators}

We begin by defining invertible synchronous Mealy machines.

\begin{definition}[Invertible synchronous Mealy machine]
	An \emph{invertible synchronous Mealy machine} is a tuple $M=(Q,\delta,\lambda,q_0)$ where $Q$ is a finite
	nonempty set of states, $\delta:Q\times\Sigma\to Q$ is a transition function, and
	$\lambda:Q\times\Sigma\to \Sigma$ is an output function such that for every state $q\in Q$, the map
	$a\mapsto \lambda(q,a)$ is a permutation of $\Sigma$. 
	
	The induced map $M:\Sigma^*\to\Sigma^*$ is defined by reading left-to-right: set $M(\lambda)=\lambda$, and for
	$w=w_1\cdots w_n$ define $u=M(w)=u_1\cdots u_n$ by the recursion
	\[
	q_{i}=\delta(q_{i-1},w_i),\qquad u_i=\lambda(q_{i-1},w_i)\quad (i=1,\dots,n),
	\]
	with $q_0$ as the initial state. In particular, $\len{M(w)}=\len{w}$ for all $w$.
\end{definition}

We state the basic properties of invertible synchronous Mealy machines.

\begin{lemma}\label{lem:mealy-basic}
	Let $M$ be an invertible synchronous Mealy machine.
	\begin{enumerate}
		\item For every $n$, the restriction $M:\Sigma^n\to\Sigma^n$ is a bijection.
		\item There exists an invertible synchronous Mealy machine $M^{-1}$ such that for all $w\in\Sigma^*$, we have $M^{-1}(M(w))=w$ .
		\item $M$ extends letter-by-letter (via the same recursion as in the definition of $M$) to a bijection $M:\Sigma^\infty\to\Sigma^\infty$ with inverse $M^{-1}$.
	\end{enumerate}
\end{lemma}
\begin{proof}
	(1) Fix $n$. Because at each step the output letter is obtained by applying a permutation depending on the current state,
	distinct inputs cannot merge: if $w\neq w'$ then at the first position $i$ where they differ, the machine is in the same state
	(because it has read the same prefix) and applies a permutation to two different letters, hence outputs different letters at position $i$.
	Thus $M$ is injective on $\Sigma^n$, hence bijective because $\Sigma^n$ is finite.
	
	(2) One constructs $M^{-1}$ by reversing the per-state permutations: in state $q$ output $\lambda(q,\cdot)^{-1}(a)$ on input $a$,
	and update the state consistently (standard Mealy-machine inversion). Because all per-state maps are permutations, this is well-defined.
	
	(3) The same recursion as in definition of $M$ works on infinite inputs; bijectivity follows from (1) on all finite prefixes.
\end{proof}

Now we formalize the finite-state coherence condition, which captures length-preserving, bounded-memory
relabelings of the standard base-$k$ grid.

\begin{definition}[Finite-state coherent enumerator]\label{def:fs-coherent}
	A function $f:\Sigma^*\to[0,1)$ is \emph{finite-state coherent} if there exists an invertible synchronous Mealy machine $M$
	such that for every nonempty $w\in\Sigma^n$, $f(w)=\mathrm{grid}(M(w))=\frac{\mathrm{val}(M(w))}{k^n}$.
	(For definiteness, set $f(\lambda)=0$.)
\end{definition}

Finite-state coherence immediately forces $f$ to enumerate the entire standard $k$-adic grid at each length.

\begin{lemma}\label{lem:fscoherent-se}
	If $f$ is finite-state coherent, then $\mathrm{Im}(f)=\bigcup_{n\ge 1}\{j/k^n:0\le j<k^n\}$, hence $f$ is a separator enumerator.
\end{lemma}
\begin{proof}
	Fix $n\ge 1$. By Lemma~\ref{lem:mealy-basic}(1), $M(\Sigma^n)=\Sigma^n$. Therefore
	\[
	f(\Sigma^n)=\mathrm{grid}(M(\Sigma^n))=\mathrm{grid}(\Sigma^n)=\left\{\frac{j}{k^n}:0\le j<k^n\right\}.
	\]
	Taking the union over $n$ gives the claimed image, which is countable and dense in $[0,1)$.
\end{proof}

For finite-state coherent enumerators, the best-from-below approximation sequence coincides with the usual base-$k$
truncations.

\begin{lemma}\label{lem:fscoherent-trunc}
	Let $f$ be finite-state coherent and $x\in[0,1)$. Then for every $n\ge 1$,
	$a_n^f(x)=\frac{\lfloor k^n x\rfloor}{k^n}$. In particular, the scaled sequence is integer-valued:
	$k^n a_n^f(x)=\lfloor k^n x\rfloor\in\{0,1,\dots,k^n-1\}$.
\end{lemma}
\begin{proof}
	By Lemma~\ref{lem:fscoherent-se}, for each $m\le n$ the set $f(\Sigma^m)$ equals the full grid
	$\{j/k^m:0\le j<k^m\}$. Hence the set of all values $f(w)$ with $\len{w}\le n$ is exactly
	$\bigcup_{m=1}^n \{j/k^m:0\le j<k^m\}$. Among the level-$n$ grid points $\{j/k^n\}$, the largest one $\le x$ is
	$\lfloor k^n x\rfloor/k^n$. It remains to check that no coarser grid point (denominator $k^m$ with $m<n$) can exceed
	this value while staying $\le x$. But for each $m<n$, $\frac{\lfloor k^m x\rfloor}{k^m}\le \frac{\lfloor k^n x\rfloor}{k^n}$, because multiplying both sides by
	$k^n$ gives $k^{n-m}\lfloor k^m x\rfloor \le \lfloor k^n x\rfloor$, which holds since
	$k^{n-m}\lfloor k^m x\rfloor \le k^n x$ and the left-hand side is an integer.
	
\end{proof}

Next we relate approximation complexity under a finite-state coherent $f$ to approximation complexity under the
standard base-$k$ enumerator. Recall that $f_{\mathrm{std}}:\Sigma^*\to[0,1)$ is the standard base-$k$ enumerator.

\begin{lemma}\label{lem:fscoherent-relabel}
	Let $f$ be finite-state coherent via a Mealy machine $M$, i.e.\ $f(w)=\mathrm{grid}(M(w))$ for all $w$.
	Then for every $\Sigma$-FST $T$, every $x\in[0,1)$, and every $\delta>0$,
	$
	K^{T,f}_{\delta}(x)=K^{M\circ T,\ f_{\mathrm{std}}}_{\delta}(x),
	$
	where $M\circ T$ denotes the output-composition transducer $\pi\mapsto M(T(\pi))$.
\end{lemma}
\begin{proof}
	Since $\mathrm{grid}(u)=f_{\mathrm{std}}(u)$ for every $u \in \Sigma^*$, we obtain 
	$K^{T,f}_{\delta}(x)=\min\{K^T(w):\, \abs{f_{\mathrm{std}}(M(w))-x}<\delta\}$.
	Substitute $u=M(w)$. Since $M:\Sigma^{\len{w}}\to\Sigma^{\len{w}}$ is bijective for each length, this is equivalent to
	$K^{T,f}_{\delta}(x)=\min\{K^T(M^{-1}(u)):\, \abs{f_{\mathrm{std}}(u)-x}<\delta\}$.
	For every $u\in\Sigma^*$,
	$K^T(M^{-1}(u))=\min\{\len{\pi}:\ T(\pi)=M^{-1}(u)\}
	=\min\{\len{\pi}:\ M(T(\pi))=u\}=K^{M\circ T}(u)$,
	and substituting yields the claim.
\end{proof}

The next proposition formalizes the key robustness property of finite-state coherence: composing the naming map with an invertible synchronous Mealy relabeling does not change the relativized finite-state approximation complexity, and hence does not change the induced $f$-dimension.

\begin{proposition}\label{prop:fscoherent-invariant}
	If $f$ is finite-state coherent, then for every $x\in[0,1)$,
	$
	\dim_{\mathrm{FS}}^{f}(x)=\dim_{\mathrm{FS}}^{f_{\mathrm{std}}}(x).
	$
	In particular, $x$ is $f$-normal if and only if $x$ is $f_{\mathrm{std}}$-normal.
\end{proposition}
\begin{proof}
	Let $f$ be finite-state coherent via $M$. By Lemma~\ref{lem:fscoherent-relabel}, for every $T$ and every $\delta$,
	$K^{T,f}_{\delta}(x)=K^{M\circ T,f_{\mathrm{std}}}_{\delta}(x)$, hence the corresponding $\liminf$ ratios coincide.
	Taking $\inf_T$ over all FSTs on the left equals taking $\inf_S$ over all FSTs on the right, because
	$T\mapsto M\circ T$ is a bijection on FSTs (with inverse $S\mapsto M^{-1}\circ S$). Therefore the two infima coincide.
\end{proof}

The next theorem states that, for the \emph{standard} base-$k$ naming map $f_{\mathrm{std}}$, the paper's notion of
$f$-normality (i.e.\ $\dim_{\mathrm{FS}}^{f}(x)=1$) coincides exactly with the classical notion of base-$k$ normality of $x$.

\begin{theorem}\label{thm:std-normal-known}
	For $x\in[0,1)$, one has $\dim_{\mathrm{FS}}^{f_{\mathrm{std}}}(x)=1$ if and only if $x$ is base-$k$ normal.
\end{theorem}

\begin{proof}
	By Theorem~3.3 of~\cite{Mayordomo2025}, for every $x\in[0,1)$,
	$
	\dim_{\mathrm{FS}}^{f_{\mathrm{std}}}(x)=\dim_{\mathrm{FS}}(\seq_k(x)).
	$
	By the standard characterization of normality via finite-state dimension (e.g.\ \cite{Bourke2005}), one has
	$\dim_{\mathrm{FS}}(\seq_k(x))=1$ if and only if $\seq_k(x)$ is base-$k$ normal. Finally, by definition, $\seq_k(x)$ is base-$k$ normal if and only if $x$ is base-$k$ normal. Combining these equivalences yields the claim.
\end{proof}

It is straightforward to verify that $k$-adic equidistribution of the sequence $(\lfloor k^n x\rfloor)_{n \ge 1}$ coincides with base-$k$ normality \cite{KuipersNiederreiterUniform}.

\begin{theorem}\label{thm:kadic-known}
	Let $x\in[0,1)$ and set $b_n=\lfloor k^n x\rfloor$. Then $x$ is base-$k$ normal if and only if $(b_n)_{n\ge 1}$ is
	$k$-adically equidistributed.
\end{theorem}

Finally we combine finite-state coherent invariance with the explicit form of $a_n^f(x)$ to obtain the
equidistribution characterization in terms of the scaled approximation sequence.

\begin{theorem}\label{thm:fscoherent-eqchar}
	Let $f$ be a finite-state coherent separator enumerator over $\Sigma$ and let $x\in[0,1)$.
	Then $x$ is $f$-normal if and only if the integer sequence $(k^n a_n^f(x))_{n\ge 1}$ is $k$-adically equidistributed.
	
\end{theorem}

\begin{proof}
	Let $b_n(x):=k^n a_n^f(x)$. Since $f$ is finite-state coherent, Proposition~\ref{prop:fscoherent-invariant} gives
	$x$ is $f$-normal if and only if $x$ is $f_{\mathrm{std}}$-normal. By Theorem~\ref{thm:std-normal-known},
	$f_{\mathrm{std}}$-normality is equivalent to base-$k$ normality of $x$. On the other hand,
	Theorem~\ref{thm:kadic-known} states that base-$k$ normality of $x$ is equivalent to $k$-adic equidistribution of the
	integer sequence $\bigl(\lfloor k^n x\rfloor\bigr)_{n\ge 1}$. Finally, Lemma~\ref{lem:fscoherent-trunc} identifies the
	scaled approximation sequence for finite-state coherent $f$ with this canonical sequence, namely
	$b_n(x)=\lfloor k^n x\rfloor$ for all $n\ge 1$. Substituting this identity into the previous equivalence yields
	$x$ is $f$-normal if and only if $(b_n(x))_{n\ge 1}$ is $k$-adically equidistributed, which is exactly the claim.
\end{proof}

\subsection{Beyond finite-state coherence: a nearly linear time computable counterexample}
\label{subsec:nlt-border}

This subsection shows that once one moves beyond finite-state coherent enumerators, a \emph{$k$-adic equidistribution}-based
criterion for $f$-normality can fail even under very low resource bounds. Concretely, we construct a nearly linear time computable separator
enumerator $f$ and a point $x\in[0,1)$ such that
$b_n(x):=k^n a_n^f(x)$ is $k$-adically equidistributed, yet $\dim_{\mathrm{FS}}^{f}(x)=0$.

Towards defining the enumerator, fix $\Sigma=\{0,1,\dots,k-1\}$ with $k=\len{\Sigma}\ge 2$ and set $x:=1/k$. For $n\ge 1$ define
$J_n:=k^{n-1}-n$ and $c_n:=J_n/k^n=1/k-n/k^n$. Then $0\le c_n<x$ for all $n$ and $(c_n)$ is nondecreasing since
$c_{n+1}-c_n=((k-1)n-1)/k^{n+1}\ge 0$.

\smallskip
\noindent\textbf{Definition of the enumerator.}
Define $f:\Sigma^*\to[0,1)$ as follows, where $n=\len{w}$.
\[
f(w):=
\begin{cases}
	c_n, & \text{if } w=0^n,\\[0.6ex]
	\frac{\mathrm{val}(w)}{k^n}, & \text{if } w \text{ starts with }0,\ w\neq 0^n,\ \mathrm{val}(w)\le J_n,\\[0.6ex]
	0, & \text{if } w \text{ starts with }0,\ w\neq 0^n,\ \mathrm{val}(w)> J_n,\\[0.6ex]
	x+\frac{1}{k^{n+1}}, & \text{if } w=1\,0^{n-1},\\[0.6ex]
	\frac{\mathrm{val}(w)}{k^n}, & \text{otherwise.}
\end{cases}
\]
In particular, clauses with leading digit $0$ satisfy $f(w)\le c_n<x$, while clauses with leading digit $\neq 0$ satisfy
$f(w)>x$.

\begin{lemma}\label{lem:pushdown-se}
	The function $f$ is a rational-valued separator enumerator, and $f$ is $O(n \log n)$ time computable.
\end{lemma}
\begin{proof}
	For density in $[0,x)$, fix $n\ge 1$. For each integer $1\le j\le J_n$ there exists a length-$n$ word
	$w$ beginning with $0$ with $\mathrm{val}(w)=j$, and then $f(w)=j/k^n$ by clause (ii). Also $0\in\mathrm{Im}(f)$:
	for any $n$ choose a length-$n$ word $w$ beginning with $0$ with $\mathrm{val}(w)>J_n$ (e.g.\ $w=0\, (k\!-\!1)^{n-1}$),
	so clause (iii) gives $f(w)=0$. Since $c_n=J_n/k^n\uparrow x$, it follows that $\bigcup_{n\ge 1}\{j/k^n:0\le j\le J_n\}$ is dense in $[0,x)$.
	
	For density in $(x,1)$, clause (v) realizes all grid points $\mathrm{val}(w)/k^n$ with first digit nonzero,
	except that the single point $x=\mathrm{val}(1\,0^{n-1})/k^n$ is replaced by $x+1/k^{n+1}>x$ by clause (iv).
	To see that the realized points are dense in $(x,1)$, fix $\varepsilon>0$ and choose $n\ge 2$ with $k^{-n}<\varepsilon$.
	Then $w:=1\,0^{n-2}\,1$ has length $n$, begins with $1$, and satisfies
	$f(w)=\mathrm{val}(w)/k^n=x+1/k^n\in(x,x+\varepsilon)$, so every interval $(x,x+\varepsilon)$ meets $\mathrm{Im}(f)$.
	Since the grid points are dense in $(x,1)$ and we only removed/perturbed one point on each level, it follows that
	$\mathrm{Im}(f)$ meets every nonempty open interval in $(x,1)$. Thus $\mathrm{Im}(f)$ is dense in $[0,1)$, so $f$ is a separator enumerator.

	We use the log-cost RAM model (see \cite{Lutz2021}), in which registers store nonnegative integers in binary and each instruction has cost proportional to the bit-lengths of the operands and addresses used; in particular, random access to the $i$-th input symbol costs $O(\log i)$. We represent the output rational $f(w)$ as a pair $(p,m)$ meaning $p/k^m$, where $p$ is written explicitly in base $k$ as a digit string (so the output length is $O(\len{w})$). In this model, $f$ is computable in time $O(\len{w}\log \len{w})$ by a direct case analysis. One first reads the input via accesses to $b_1,\dots,b_{\len{w}+1}$ in order to determine $\len{w}$ and a constant number of flags (whether $w=0^{\len{w}}$, whether $w=1\,0^{\len{w}-1}$, and whether the first digit is $0$); this costs $\sum_{i\le \len{w}}O(\log i)=O(\len{w}\log \len{w})$, and the auxiliary bookkeeping uses only $O(\log \len{w})$-bit counters. In the easy cases, the output $(p,m)$ is produced by copying $w$ or by writing the required base-$k$ digit string for $p$ (namely, the explicit strings for $k^{\len{w}}+1$ or $J_{\len{w}}$), which takes $O(\len{w})$ digit steps and hence $O(\len{w}\log \len{w})$ time accounting for log-cost indexing/output. In the remaining case $w=0u$ with $\len{u}=\len{w}-1$, one computes the base-$k$ digit string of $J_{\len{w}}=k^{\len{w}-1}-\len{w}$ and compares it lexicographically with $u$; both tasks take $O(\len{w})$ digit steps, each incurring only $O(\log \len{w})$ overhead from log-cost addressing/counters. Thus the total running time is $O(\len{w}\log \len{w})$.
\end{proof}

We next identify the induced best-from-below approximation chain at the special point $x=1/k$.

\begin{lemma}\label{lem:pushdown-an}
	For every $n\ge 1$, one has $a_n^f(x)=c_n$. Consequently,
	$
	b_n(x)=k^n a_n^f(x)=J_n=k^{n-1}-n.
	$
\end{lemma}
\begin{proof}
	Since $f(0^n)=c_n$ by clause (i) and $c_n=1/k-n/k^n<x=1/k$, we have $a_n^f(x)\ge c_n$. Now let $w\in\Sigma^n$ with $f(w)<x$. If $w$ begins with a nonzero digit, then $f(w)>x$ by clause (iv) or (v), a contradiction; hence $w$ begins with $0$. If $w=0^n$ then $f(w)=c_n$. If $w\neq 0^n$, then by clause (ii) or (iii) either $f(w)=\mathrm{val}(w)/k^n$ with $\mathrm{val}(w)\le J_n$, or $f(w)=0$; in the first case
	\[
	f(w)=\frac{\mathrm{val}(w)}{k^n}\le \frac{J_n}{k^n}=c_n,
	\]
	and in the second case $f(w)=0\le c_n$. Thus every $w\in\Sigma^n$ with $f(w)<x$ satisfies $f(w)\le c_n$, so $a_n^f(x)\le c_n$. Therefore $a_n^f(x)=c_n$, and multiplying by $k^n$ yields $b_n(x)=k^n c_n=J_n=k^{n-1}-n$.
\end{proof}

We now verify that the sequence $(b_n(x))_{n\ge 1}$ is $k$-adically equidistributed.

\begin{lemma}\label{lem:pushdown-kadic}
	The integer sequence $(b_n(x))_{n\ge 1}$ is $k$-adically equidistributed.
\end{lemma}

\begin{proof}
	Fix $m\ge 1$. For all $n\ge m+1$ one has $k^m\mid k^{n-1}$, hence
	$b_n(x)=k^{n-1}-n\equiv -n \pmod{k^m}$ for all $n\ge m+1$. As $n$ ranges over $\{1,2,\dots,N\}$,
	the residues of $-n\bmod k^m$ are asymptotically uniform, and omitting the finite initial segment $n\le m$
	does not affect limiting frequencies. Therefore, for every $r\in\{0,1,\dots,k^m-1\}$,
	$\lim_{N\to\infty}\frac{1}{N}\#\{1\le n\le N:\ b_n(x)\equiv r\!\!\!\pmod{k^m}\}=\frac{1}{k^m}$.
\end{proof}

Finally, despite this equidistribution property, the point $x$ has vanishing relativized finite-state dimension with respect to $f$.

\begin{proposition}\label{prop:pushdown-dim0}
	For the above $f$ and $x=1/k$, one has $\dim_{\mathrm{FS}}^{f}(x)=0$.
\end{proposition}

\begin{proof}
	For each $L\ge 1$ let $T_L$ be the $1$-state $\Sigma$-FST that outputs $0^L$ on every input symbol. Then
	$T_L(\pi)=0^{L\len{\pi}}$, so $K^{T_L}(0^n)\le \lceil n/L\rceil$. Let $\delta_n:=2n/k^n$. Since $f(0^n)=c_n$ and $|x-c_n|=\frac{n}{k^n}<\delta_n$ we have
	$K^{T_L,f}_{\delta_n}(x)\le K^{T_L}(0^n)\le \lceil n/L\rceil$. Since
	$\log_k(1/\delta_n)=\log_k(k^n/(2n))=n-\log_k(2n)$, we obtain
	$$
	\liminf_{n\to\infty}\frac{K^{T_L,f}_{\delta_n}(x)}{\log_k(1/\delta_n)}
	\le
	\liminf_{n\to\infty}\frac{\lceil n/L\rceil}{n-\log_k(2n)}
	=\frac{1}{L}.
	$$
	Taking the infimum over transducers (in particular over $T_L$) and letting $L\to\infty$ yields
	$\dim_{\mathrm{FS}}^{f}(x)=0$.
\end{proof}

Taking $x=1/k$ and $f$ as defined above we obtain the following.

\begin{theorem}\label{thm:pushdown-summary}
	There exist a rational-valued $O(n \log n)$-time computable separator enumerator $f:\Sigma^*\to[0,1)$ and a point $x\in[0,1)$ such that the integer sequence $(k^n a_n^f(x))_{n\ge 1}$ is $k$-adically equidistributed,
	while $\dim_{\mathrm{FS}}^{f}(x)=0$.
\end{theorem}

\begin{proof}
	Take $x=1/k$ and $f$ as defined above. Lemma~\ref{lem:pushdown-an} identifies $k^n a_n^f(x)=b_n(x)$, Lemma~\ref{lem:pushdown-kadic}
	shows $(b_n(x))$ is $k$-adically equidistributed, and Proposition~\ref{prop:pushdown-dim0} gives $\dim_{\mathrm{FS}}^{f}(x)=0$.
\end{proof}

\section{Discussion and open questions}
\label{sec:discussion}

We show that, in general, no distributional property of the scaled best-from-below approximation sequence
$(k^n a_n^f(x))_{n\ge1}$ can characterize $f$-normality uniformly over all separator enumerators (indeed, we construct computable
enumerators $f_0,f_1$ for which the associated scaled approximation sequences coincide while the corresponding
$f$-normality behavior diverges). At the same time, we identify a structured regime---finite-state coherent enumerators---in which $f$-normality is
equivalent to $k$-adic equidistribution of $(k^n a_n^f(x))_{n\ge1}$. We also show that this correspondence can already fail for an efficiently computable separator enumerator, in fact for an $O(n\log n)$-time computable $f$. This raises the natural question of how far this correspondence persists beyond finite-state coherent relabelings.
One concrete setting is when the naming map is computable by a deterministic pushdown transducer.
In particular, does there exist such a separator enumerator $f$ and a point $x\in[0,1)$ for which
the integer sequence $(k^n a_n^f(x))_{n\ge 1}$ is $k$-adically equidistributed while $\dim_{\mathrm{FS}}^{f}(x)<1$? It would also be interesting to identify the widest natural family of separator enumerators for which an equidistribution characterization of $f$-normality remains valid.

A related direction concerns the role of invertibility in the defining Mealy-machine relabeling: invertibility is used
to transport approximation complexity via the substitution $u=M(w)$ and to ensure that $T\mapsto M\circ T$ is a bijection
on finite-state transducers, yielding invariance of the induced $f$-dimension. It would be interesting to determine
whether some weaker condition (e.g.\ levelwise surjectivity or bounded-to-one behavior on each $\Sigma^n$) suffices for
the same equidistribution characterization, or whether non-invertible finite-state relabelings can already break it.

%
%
%

\bibliography{main}

@InProceedings{Mayordomo2025,
author="Mayordomo, Elvira",
editor="Beckmann, Arnold
and Oitavem, Isabel
and Manea, Florin",
title="A Point to Set Principle for Finite-State Dimension",
booktitle="Crossroads of Computability and Logic: Insights, Inspirations, and Innovations",
year="2025",
publisher="Springer Nature Switzerland",
address="Cham",
pages="299--304",
isbn="978-3-031-95908-0"
}

@misc{Doty2006,
      title={Finite-State Dimension and Lossy Decompressors}, 
      author={David Doty and Philippe Moser},
      year={2006},
      eprint={cs/0609096},
      archivePrefix={arXiv},
      primaryClass={cs.CC},
      url={https://arxiv.org/abs/cs/0609096}, 
}

@article{Bourke2005,
  title={Entropy rates and finite-state dimension},
  author={Bourke, Chris and Hitchcock, John M and Vinodchandran, NV},
  journal={Theoretical Computer Science},
  volume={349},
  number={3},
  pages={392--406},
  year={2005},
  publisher={Elsevier}
}

@article{Dai2004,
  title={Finite-state dimension},
  author={Dai, Jack J and Lathrop, James I and Lutz, Jack H and Mayordomo, Elvira},
  journal={Theoretical Computer Science},
  volume={310},
  number={1-3},
  pages={1--33},
  year={2004},
  publisher={Elsevier}
}

@article{Lutz2003a,
title = {The dimensions of individual strings and sequences},
journal = {Information and Computation},
volume = {187},
number = {1},
pages = {49-79},
year = {2003},
issn = {0890-5401},
doi = {https://doi.org/10.1016/S0890-5401(03)00187-1},
url = {https://www.sciencedirect.com/science/article/pii/S0890540103001871},
author = {Jack H. Lutz}
}

@article{Lutz2003b,
author = {Lutz, Jack H.},
title = {Dimension in Complexity Classes},
journal = {SIAM Journal on Computing},
volume = {32},
number = {5},
pages = {1236-1259},
year = {2003},
doi = {10.1137/S0097539701417723},

URL = { 
    
        https://doi.org/10.1137/S0097539701417723
    
    

},
eprint = { 
    
        https://doi.org/10.1137/S0097539701417723
    
    

}
}

@article{Lutz2018,
author = {Lutz, Jack H. and Lutz, Neil},
title = {Algorithmic Information, Plane Kakeya Sets, and Conditional Dimension},
year = {2018},
issue_date = {June 2018},
publisher = {Association for Computing Machinery},
address = {New York, NY, USA},
volume = {10},
number = {2},
issn = {1942-3454},
url = {https://doi.org/10.1145/3201783},
doi = {10.1145/3201783},
journal = {ACM Trans. Comput. Theory},
month = may,
articleno = {7},
numpages = {22}
}

@book {KuipersNiederreiterUniform,
    AUTHOR = {Kuipers, L. and Niederreiter, H.},
     TITLE = {Uniform distribution of sequences},
    SERIES = {Pure and Applied Mathematics},
 PUBLISHER = {Wiley-Interscience [John Wiley \& Sons], New
              York-London-Sydney},
      YEAR = {1974},
     PAGES = {xiv+390},
   MRCLASS = {10K05 (22D99)},
  MRNUMBER = {0419394},
MRREVIEWER = {P. Gerl},
}

@article{Calvert2025,
  title={Normality, relativization, and randomness},
  author={Calvert, Wesley and Gruner, Emma and Mayordomo, Elvira and Turetsky, Daniel and Villano, Java Darleen},
  journal={Theory of Computing Systems},
  volume={69},
  number={3},
  pages={26},
  year={2025},
  publisher={Springer}
}

@incollection{Lutz2020,
  title={Who asked us? How the theory of computing answers questions about analysis},
  author={Lutz, Jack H and Lutz, Neil},
  booktitle={Complexity and Approximation: In Memory of Ker-I Ko},
  pages={48--56},
  year={2020},
  publisher={Springer}
}

@article{Lutz2021,
title = {Computing absolutely normal numbers in nearly linear time},
journal = {Information and Computation},
volume = {281},
pages = {104746},
year = {2021},
issn = {0890-5401},
doi = {https://doi.org/10.1016/j.ic.2021.104746},
url = {https://www.sciencedirect.com/science/article/pii/S0890540121000614},
author = {Jack H. Lutz and Elvira Mayordomo}
}
\bibliographystyle{plain}

\end{document}